\documentclass{article}
\usepackage{arxiv}

\usepackage[utf8]{inputenc} % allow utf-8 input
\usepackage[T1]{fontenc}    % use 8-bit T1 fonts
\usepackage{hyperref}       % hyperlinks
\usepackage{url}            % simple URL typesetting
\usepackage{booktabs}       % professional-quality tables
\usepackage{amsfonts}       % blackboard math symbols
\usepackage{nicefrac}       % compact symbols for 1/2, etc.
\usepackage{microtype}      % microtypography
\usepackage{lipsum}

\usepackage{times}
\usepackage{soul}
\usepackage{url}
\usepackage[small]{caption}
\usepackage{graphicx}
\usepackage{amsmath}
\usepackage{amsthm}
\usepackage{booktabs}
\usepackage{algorithm}
\usepackage{algorithmic}
\usepackage{todonotes}
\usepackage{mathtools}
\usepackage{multirow}

\usepackage{amssymb}
\usepackage{natbib}
\usepackage{enumitem}
\usepackage{hyperref}

\newtheorem{example}{Example}
\newtheorem{theorem}{Theorem}
\newtheorem{lemma}{Lemma}

\newtheorem{definition}{Definition}

\newcommand{\bigtimes}{\text{\LARGE $\times$}}
\newcommand{\NPHard}{$\mathsf{NP}$-hard}

\newcommand{\defeq}{\vcentcolon=}

\newcommand{\strA}{\textsf{A}}
\newcommand{\strB}{\textsf{B}}

\DeclareMathOperator*{\argmin}{argmin}

\title{Signaling in Bayesian Network Congestion Games:\\ the Subtle Power of Symmetry}

\author{
	Matteo Castiglioni\\
	Politecnico di Milano\\
	\texttt{matteo.castiglioni@polimi.it}
	\And
	Andrea Celli\\
	Politecnico di Milano\\
	\texttt{andrea.celli@polimi.it}
	\And
	Alberto Marchesi\\
	Politecnico di Milano\\
	\texttt{alberto.marchesi@polimi.it}
	\And
	Nicola Gatti\\
	Politecnico di Milano\\
	\texttt{nicola.gatti@polimi.it}
}

\begin{document}

\maketitle

\begin{abstract}
% setting
Network congestion games are a well-understood model of multi-agent strategic interactions.
Despite their ubiquitous applications, it is not clear whether it is possible to design information structures to ameliorate the overall {experience} of the network users.
We focus on Bayesian games with atomic players, where network vagaries are modeled via a (random) state of nature which determines the costs incurred by the players.
A third-party entity---the {\em sender}---can observe the realized state of the network and exploit this additional information to send a {\em signal} to each player.
A natural question is the following: {\em is it possible for an informed sender to reduce the overall social cost via the strategic provision of information to players who update their beliefs rationally?}
The paper focuses on the problem of computing optimal {\em ex ante} persuasive signaling schemes, showing that {\em symmetry} is a crucial property for its solution.
%
% OLD :::  Indeed, we show that it is possible to compute optimal {\em ex ante} persuasive signaling schemes when players are symmetric and have affine cost functions. 
%
% OLD :::  Moreover, the problem becomes \NPHard~when players are asymmetric, even in a non-Bayesian setting.
%
Indeed, we show that an optimal {\em ex ante} persuasive signaling scheme can be computed in polynomial time when players are symmetric and have affine cost functions.
Moreover, the problem becomes \NPHard~when players are asymmetric, even in non-Bayesian settings.
\end{abstract}

\section{Introduction}\label{sec:intro}

{\em Network congestion games}, where players seek to minimize their own costs selfishly, are a canonical example of a setting where externalities may induce socially inefficient outcomes~\citep{roughgarden2005selfish}.
In real-world problems, the state of the network may be uncertain, and not known to its users ({\em e.g.}, drivers may not be aware of road works and accidents in a road network).
This setting is modeled via {\em Bayesian network congestion games} (BNCGs). 
We investigate whether providing players with partial information about the state of the network may mitigate inefficiencies.

% bayesian persuasion + applicability
We model this information-structure design problem through the Bayesian persuasion framework by~\cite{kamenica2011bayesian}.
At its core, this framework involves an informed {\em sender} trying to influence the behavior of a set of self-interested players---the {\em receivers}---via the provision of payoff-relevant information.
We focus on the notion of {\em ex ante persuasiveness}, as introduced by~\cite{xu2019tractability} and~\cite{celli2019bayesian}, where players commit to following the sender's recommendations having observed only the information structure.
This assumes credible receivers' commitments, which is reasonable in practice.
In our setting, where signaling schemes are usually implemented as software ({\em e.g.}, real-time traffic apps), it is natural to assume that each player decides to either follow the signaling scheme ({\em i.e.}, adopting the software) or act based on his prior belief.
Moreover, on a general level, the receivers will uphold their {\em ex ante} commitment every time they reason with a long-term horizon where a reputation for credibility positively affects their utility~\citep{rayo2010optimal}.
In some cases, the receivers could also be forced to stick to their {\em ex ante} commitment by contractual agreements or penalties.

\paragraph{Related Works}
~\cite{arnott1991does} and~\cite{acemoglu2018informational} study the impact of information on traffic congestion.
Several recent works focus on {\em non-atomic} games~\citep{das2017reducing,massicot2019public,wu2018value,vasserman2015implementing}.
\cite{bhaskar2016hardness} study the inapproximability of finding optimal {\em ex interim} persuasive signaling schemes in non-atomic games.
\cite{liu2019efficient} focus on atomic games with costs uncertainties and study {\em ex interim} persuasion by placing stringent constraints on the network structure.
To the best of our knowledge, the present work is the first studying {\em ex ante} persuasion in general atomic BNCGs.
Other works study the problem of correlation in non-Bayesian congestion games \citep{christodoulou2005,papadimitriou2008}. 
The closest to our work is that of~\cite{jiang2011}, who provide a polynomial-time algorithm to find an optimal coarse correlated equilibrium ({\em i.e.}, an {\em ex ante} persuasive signaling scheme in the non-Bayesian setting) in simple games with symmetric players selecting a single resource (a.k.a. \emph{singleton} congestion games).

\paragraph{Original Contributions}
We investigate whether it is possible to efficiently compute optimal ({\em i.e.}, minimizing the social cost) {\em ex ante} persuasive signaling schemes in BNCGs.
First, we show that an optimal {\em ex ante} persuasive signaling scheme can be compute in poly-time in symmetric BNCGs with affine costs.
To prove this result, we exploit the ellipsoid algorithm by designing a sophisticated polynomial-time separation oracle based on a suitably defined min-cost flow problem.
Then, we show that {\em symmetry} is a crucial property for efficient signaling by proving that it is \NPHard~to compute an optimal {\em ex ante} persuasive signaling scheme in {\em asymmetric} BNCGs.
Our reduction proves an even stronger hardness result, as it works for non-Bayesian singleton congestion games with affine costs, which is arguably the simplest class of asymmetric congestion games.
% OLD :::  Our reduction proves an even stronger hardness result, as it works for non-Bayesian singleton congestion games with affine costs, which is arguably the simplest class of congestion games where asymmetry can be introduced.
%
% OLD :::  Thus, our reduction also shows that finding an optimal coarse correlated equilibrium in these games is an \NPHard~problem.
Furthermore, in such setting, a solution to our problem is an optimal coarse correlated equilibrium and, thus, its computation is \NPHard.

%\red{
%Our reduction proves an even stronger hardness result, as it is based on non-Bayesian asymmetric singleton congestion games (\emph{i.e.}, where each player selects a single resource), which is arguably the simplest class of congestion games with asymmetric players. 
%%
%This shows that finding an optimal coarse correlated equilibrium in these games is \NPHard.
%}

%\clearpage

% \input{content/prelim}

\section{Signaling in Network Congestion Games}

We study atomic network congestion games where edge costs depend on a stochastic state of nature, defined as follows.
%
%In this section, we introduce the main elements of our model.

\paragraph{Network Congestion Game (NCG)}

A \emph{network congestion game}~\citep{fabrikant2004complexity} is defined as a tuple $(N,G,\{c_e\}_{e \in E},\{(s_p,t_p)\}_{p \in N})$, where:
$N \defeq \{1, \ldots, n\}$ denotes the set of players;
$G \defeq (V,E)$ is the directed graph underlying the game, with $V$ being its set of nodes and each $e = (v,v') \in E$ representing a directed edge from $v$ to $v'$;
$\{c_e\}_{e \in E}$ are the edge costs, with each $c_e : \mathbb{N} \to \mathbb{R}_+$ defining the cost of edge $e \in E$ as a function of the number of players traveling through $e$;
finally, $\{(s_p,t_p)\}_{p \in N}$, with $s_p, t_p \in V$, denote the source-destination pairs for all the players.
As usual, we assume that $c_e(0) = 0$ for all $e \in E$.
In an NCG, the set $A_p$ of actions available to a player $p \in N$ is implicitly defined by the graph $G$, the source $s_p$, and the destination $t_p$.
Formally, $A_p$ is the set of all directed paths from $s_p$ to $t_p$ in the graph $G$.
In this work, we use $a_p \in A_p$ to denote a player $p$'s path and we write $e \in a_p$ whenever the path contains the edge $e \in E$.
An action profile $a \in A$, where $A \defeq \bigtimes_{p \in N} A_p$, is a tuple of $s_p$-$t_p$ directed paths $a_p \in A_p$, one per player $p \in N$.
For the ease of notation, given an action profile $a \in A$, we let $f_e^a$ be the congestion of edge $e \in E$ in $a$, \emph{i.e.}, the number of players selecting a path passing thorough $e$ in $a$; formally, $f_e^a \defeq |\{ p \in N \mid e \in a_p  \}|$.
Thus, $c_e(f_e^a)$ denotes the cost of edge $e$ in $a$.
%
%Finally, we let the cost incurred by player $p \in N$ in an action profile $a \in A$ be denoted as $c_p(a) \defeq \sum_{e \in a_p} c_e(f_e^a)$.
%
Finally, the cost incurred by player $p\in N$ in an action profile $a\in A$ is denoted by $c_p(a) \defeq \sum_{e \in a_p} c_e(f_e^a)$.

\paragraph{Bayesian Network Congestion Game (BNCG)}
We define a \emph{Bayesian network congestion game} as a tuple $(N,G,\Theta, \mu, \{c_{e,\theta}\}_{e \in E, \theta \in \Theta},\{(s_p,t_p)\}_{p \in N})$, where, differently from the basic setting, the edge cost functions $c_{e,\theta}: \mathbb{N} \to \mathbb{R}_+$ also depend on a state of nature $\theta$ drawn from a finite set of states $\Theta$.
Moreover, $\mu$ encodes the prior beliefs that the players have over the states of nature, \emph{i.e.}, $\mu \in \textnormal{int}(\Delta_\Theta)$ is a fully-supported probability distribution over the set $\Theta$, with $\mu_\theta$ denoting the prior probability that the state of nature is $\theta \in \Theta$.
All the other components are defined as in non-Bayesian NCGs.
Notice that, in BNCGs, the cost experienced by player $p \in N$ in an action profile $a \in A$ also depends on the state of nature $\theta \in \Theta$, and, thus, it is defined as $c_{p,\theta}(a) \defeq \sum_{e \in a_p} c_{e,\theta}(f_e^a)$.
%
% In this work, we will focus on BNCGs with two additional properties.
%
A BNCG is \emph{symmetric} if all the players share the same $(s_p, t_p)$ pair, \emph{i.e.}, whenever they all have the same set of actions (paths).
For the ease of notation, in such settings we let $s, t \in V$ be the common source and destination.
%
% Instead, if the players have different $s_p$-$t_p$ pairs, we say that the game is \emph{asymmetric}.
%
Moreover, we focus on BNCGs with \emph{affine costs}, \emph{i.e.}, for all $e \in E$ and $\theta \in \Theta$, there exist constants $\alpha_{e,\theta}, \beta_{e,\theta} \in \mathbb{R}_+$ such that the edge cost function is $c_{e,\theta}(f_e^a) \defeq \alpha_{e,\theta} f_e^a + \beta_{e,\theta}$.~\footnote{We focus on affine costs since: \emph{(i)} the assumption is reasonable in many applications~\citep{vasserman2015implementing}, and \emph{(ii)} the problem is trivially \NPHard~when generic costs are allowed (see Section~\ref{sec:hard}).}

\paragraph{Signaling in BNCGs}
Suppose that a BNCG is employed to model a road network subject to vagaries.
It is reasonable to assume that third-party entities ({\em e.g.}, the road management company) may have access to the realized state of nature.
We call one such entity {\em the sender}.
We focus on the following natural question: {\em is it possible for an informed sender to mitigate the overall costs through the strategic provision of information to players who update their beliefs rationally?}
The sender can publicly commit to a {\em signaling scheme} which maps the realized state of nature to a signal for each player. 
The sender can exploit general {\em private} signaling schemes, sending different signals to each player through private communication channels.
In this setting, a simple revelation-principle-style argument shows that it is enough to employ players' actions as signals~\citep{arieli2016private,kamenica2011bayesian}.
Therefore, a private signaling scheme is a function $\phi:\Theta\to \Delta_A$ which maps any state of nature to a probability distribution over action profiles (signals), and recommends action $a_p$ to player $p$. 
The probability of recommending an action profile $a \in A$ having observed the state of nature $\theta \in \Theta$ is denoted by $\phi_{\theta,a}$.
Then, it has to hold $\sum_{a\in A} \phi_{\theta,a}=1$, for each $\theta\in\Theta$.
A signaling scheme is {\em persuasive} if following recommendations is an equilibrium of the underlying {\em Bayesian game}~\citep{bergemann2016bayes,bergemann2016information}.
We focus on the notion of {\em ex ante persuasiveness} as defined by~\cite{xu2019tractability}~and~\cite{celli2019bayesian}.
\begin{definition}
	A signaling scheme $\phi:\Theta\to\Delta_A$ is {\em ex ante persuasive} if, for each $p\in N$ and $a_p\in A_p$, it holds:
	\[
	\sum_{\theta\in\Theta}\mu_\theta\hspace{-.2cm}\sum_{a'=(a'_p,a_{-p})\in A} \hspace{-.2cm} \phi_{\theta,a'} \Big(c_{p,\theta}(a_p,a_{-p})-c_{p,\theta}(a')\Big)\geq 0.
	\]
\end{definition}
Then, a {\em coarse correlated equilibrium} (CCE)~\citep{moulin1978strategically} may be seen as an {\em ex ante} persuasive signaling scheme when $|\Theta|=1$.
Finally, a sender's {\em optimal ex ante} persuasive signaling scheme $\phi^\ast$ is such that it minimizes the {\em expected social cost} of the solution, {\em i.e.}:
\[
\phi^\ast\in \argmin_{\phi}\sum_{\theta\in\Theta}\mu_\theta \sum_{a\in A}\phi_{\theta,a}\sum_{p\in N} c_{p,\theta}(a).
\]

The following example illustrates the interaction flow between the sender and the players (receivers).
\begin{example}\label{ex:persuasion}
	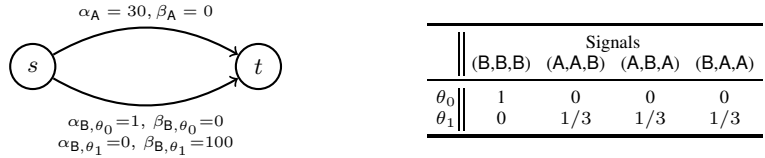
\begin{figure}[h]
	\centering
	\hspace{-1cm}
	\begin{minipage}{.2\textwidth}
	\small
			\begin{tikzpicture}
	  			\tikzstyle{node}=[circle,thick,draw=black,minimum size=6mm]
				\tikzstyle{transition}=[rectangle,thick,draw=black!75,
				fill=black!20,minimum size=4mm]
				\node [node] (s) at (0,0) {$s$};
				\node [node] (t) at (3,0) {$t$};
				\draw[->] (s) edge[thick,black,bend left=30] node[yshift=2mm] {\tiny{$\alpha_\strA=30,\beta_\strA=0$}} (t);
				\draw[->] (s) edge[thick,black,bend right=30] node[yshift=-4mm]{{$\substack{\alpha_{\strB, \theta_0} =1, \,\, \beta_{\strB, \theta_0}=0\\ \alpha_{\strB, \theta_1}=0, \,\, \beta_{\strB, \theta_1}=100}$}}(t);
			\end{tikzpicture}
	\end{minipage}
	\hspace{2cm}
	\renewcommand{\arraystretch}{1}\setlength{\tabcolsep}{1pt}
	\begin{minipage}{.2\textwidth}
		\scriptsize
		\begin{tabular}{lr||lcr lcr lcr lcr}
			\toprule
			&&&& \multicolumn{8}{c}{Signals}&\\
			&&& (\strB,\strB,\strB) &&& (\strA,\strA,\strB) &&& (\strA,\strB,\strA) &&& (\strB,\strA,\strA)\\
			\midrule
			&$\theta_0$ && 1 &&& 0 &&& 0 &&& 0 &\\
			&$\theta_1$ && 0 &&& $1/3$ &&& $1/3$ &&& $1/3$ &\\
			\bottomrule
		\end{tabular}
	\end{minipage}
	\caption{{\em Left}: BNCG for Example~\ref{ex:persuasion}. {\em Right}: An {\em ex ante} persuasive signaling scheme for the case with $n=3$. The table displays only those $a\in A$ such that $\phi_{\theta,a}>0$ for some $\theta\in\Theta = \{\theta_0, \theta_1\}$.}
	\label{fig:ex1}
	\end{figure}

Figure~\ref{fig:ex1}~(Left) describes a simple BNCG modeling the road network between the Tokyo Haneda airport (node $s$), and Yokohama (node $t$).
It is late at night and three lone researchers have to reach the IJCAI venue.
They are following navigation instructions from the same application, whose provider (the sender) has access to the current state of the roads (called \emph{\strA}~and \emph{\strB}, respectively).
Roads costs (i.e., travel times) are depicted in Figure~\ref{fig:ex1}~(Left).
In normal conditions (state $\theta_0$), road \emph{\strB}~is extremely fast ($\alpha_\emph{\strB}=1$ and $\beta_\emph{\strB}=0$).
However, it requires frequent road works for maintenance (state $\theta_1$), which increase the travel time.
Moreover, it holds $\mu_{\theta_0}=\mu_{\theta_1}=1/2$.
The interaction between the sender and the three players goes as follows: 
(i) the sender commits to a signaling scheme $\phi$;
(ii) the players observe $\phi$ and decide whether to adhere to the navigation system or not;
(iii) the sender observes the realized state of nature and exploits this knowledge to compute recommendations.
Figure~\ref{fig:ex1}~(Right) describes an {\em ex ante} persuasive signaling scheme.
In this case, when the state of nature is $\theta_1$, one of the players is randomly selected to take road \emph{\strB}, even if it is undergoing maintenance.
In expectation, following the sender's recommendations is strictly better than congesting road \emph{\strA}.  
\end{example}

A simple variation of Example~\ref{ex:persuasion} is enough to show that the introduction of signaling allows the sender to reach solutions with arbitrarily better expected social cost than what can be achieved via the optimal Bayes-Nash equilibrium in absence of signaling.
Specifically, consider the BNCG in Figure~\ref{fig:ex1}~(Left) with the following modifications:
$n=1$, $\beta$ coefficients always equal to zero, $\alpha_{\strA,\theta_0}=\infty$, $\alpha_{\strA,\theta_1}=0$, $\alpha_{\strB,\theta_0}=0$, and $\alpha_{\strB,\theta_1}=\infty$.
Without  signaling, the optimal choice yields an expected social cost of $\infty$.
However, a perfectly informative signal allows the players to avoid any cost.

\section{The Power of Symmetry}\label{sec:poly}

We design a polynomial-time algorithm to compute an optimal \emph{ex ante} persuasive signaling scheme in symmetric BNCGs with affine cost functions.
Our algorithm exploits the ellipsoid method.
We first formulate the problem as an LP (Problem~\eqref{LP:1}) with polynomially many constraints and exponentially many variables.
Then, we show how to find an optimal solution to the LP in polynomial time by applying the ellipsoid algorithm to its dual (Problem~\eqref{LP:2}), which features polynomially many variables and exponentially many constraints.
This calls for a polynomial-time separation oracle for Problem~\eqref{LP:2}, which is not readily available since the problem has an exponential number of constraints.
We prove that, in our setting, a polynomial-time separation oracle can be implemented by solving a suitably defined min-cost flow problem. 
The proof of this result crucially relies on the symmetric nature of the problem and the assumption that the costs are affine functions of the edge congestion.

The following lemma shows how to formulate the problem as an LP.~\footnote{LPs analogous to Problem~\eqref{LP:1} and Problem~\eqref{LP:2} can also be derived for the asymmetric setting. However, the separation problem for the latter is solvable in poly-time only in the symmetric case.}
%
%~\footnote{Notice that, in this section, we focus on symmetric BNCGs with a single source $s \in V$ and target $t \in V$. LPs analogous to Problem~\ref{LP:1} and its dual Problem~\ref{LP:2} can be derived also for the general, asymmetric setting. However, the separation problem associated with the dual is solvable in polynomial time only in the symmetric case.}
%
For the ease of presentation, we use $I_{\{e \notin a_p\}}$ to denote the indicator function for the event $e \notin a_p$, \emph{i.e.}, it holds $I_{\{e \notin a_p\}} = 1$ if $e \notin a_p$, while $I_{\{e \notin a_p\}}= 0$ otherwise.

\begin{lemma}
	Given a symmetric BNCG, an optimal \emph{ex ante} persuasive signaling scheme $\phi$ can be found with the LP:
	\begin{subequations}\label{LP:1}
		\begin{align}
			\min_{\phi \geq 0, x} & \,\,\, \sum_{\theta \in \Theta} \mu_\theta \sum_{a \in A} \phi_{\theta,a}  \sum_{p \in N} c_{p,\theta}(a) \label{obj:1} \\
			&\hspace{-3mm} \sum_{\theta \in \Theta} \mu_\theta  \sum_{a \in A}  c_{p,\theta}(a) \phi_{\theta,a}  \le  x_{p,s} & \forall p \in N \label{c:1}\\
			&\hspace{-3mm}x_{p,v}\le \sum_{\theta \in \Theta} \mu_\theta \sum_{a \in A} c_{e,\theta}\left(f_e^a+I_{\{e \notin a_p \}} \right) \phi_{\theta,a} + x_{p,v'} &\forall p \in N,\forall e=(v,v') \in E \label{c:2}\\
			&\hspace{-3mm}x_{p,t}=0 & \forall p \in N  \label{c:3}\\
			&\hspace{-3mm}\sum_{a \in A} \phi_{\theta,a}=1 &\forall \theta \in \Theta \label{c:4}%\\
			%&\hspace{-3mm}\phi_{\theta,a} \ge 0 \hspace{3.4cm} \forall \theta \in \Theta, \forall a \in A. \label{c:5}
		\end{align}
	\end{subequations}
\end{lemma}

\begin{proof}
	Clearly, Objective~\eqref{obj:1} is equivalent to minimizing the social cost, while Constraints~\eqref{c:4} imply that $\phi$ is well formed.
	Constraints~\eqref{c:1} enforce \emph{ex ante} persuasiveness for every player $p \in N$: the expression on the left-hand side represents player $p$'s expected cost, while $x_{p,s}$ is the cost of her best deviation (\emph{i.e.}, a cost-minimizing path given $\mu$ and $\phi$).
	This is ensured by Constraints~\eqref{c:2}~and~\eqref{c:3}.
	In particular, for every player $p \in N$ and node $v \in V \setminus \{t\}$, the former guarantee that $x_{p,v}$ is the minimum cost of a path from $v$ to $t$.
	This is shown by noticing that (given that $x_{p,t}=0$) such cost can be inductively defined as follows:
	$$
		\min_{\substack{v' \in V: \\e=(v,v') \in E}} \sum_{\theta \in \Theta} \mu_\theta \sum_{a \in A} c_{e,\theta} \left( f_e^a +I_{\{e \notin a_p \}} \right) \phi_{\theta,a} + x_{p,v'},
	$$ 
	where $f_e^a +I_{\{e \notin a_p\}}$ accounts for the fact that the congestion of edge $e$ must be incremented by one if player $p$ does not select a path containing $e$ in the action profile $a$.
\end{proof}

\begin{lemma}
	The dual of Problem~\eqref{LP:1} reads as follows:%
	\begin{subequations}\label{LP:2}
		\begin{align}
		\max_y & \,\,\, \sum_{\theta \in \Theta}  y_\theta \\
		& \hspace{-3mm} \mu_\theta \hspace{-1mm} \left( \sum_{p \in P} c_{p,\theta}(a) y_p - \hspace{-1mm}\sum_{p \in P} \sum_{e \in E} c_{e,\theta} \left( f_e^a+I_{\{e \notin a_p\}} \right) y_{p,e} \right) \hspace{-1mm} + y_\theta \le  \mu_\theta \sum_{p \in N} c_{p,\theta}(a) & \forall \theta \in \Theta, \forall a \in A \label{d:1}\\
		&\hspace{-3mm} \sum_{{v' \in V: e=(v,v') \in E}} y_{p,e} - \sum_{{v' \in V: e=(v',v) \in E}} y_{p,e} = 0 &  \forall p \in N, \forall v \in V \setminus \{s, t\}  \label{d:2} \\
		&\hspace{-3mm} \sum_{{v \in V: e=(s,v) \in E}} y_{p,e}  - \sum_{p \in N} y_p = 0 & \forall p \in N \label{d:3}\\
		& \hspace{-3mm}\sum_{p \in N} y_{p,t}- \sum_{{v \in V: e=(v,t) \in E}} y_{p,e} = 0 & \forall p \in N \label{d:4}\\
		&\hspace{-3mm} y_p \le 0 &\forall p \in N \\
		& \hspace{-3mm}y_{p,e} \le 0 & \forall p \in N, \forall e \in E.
		\end{align}
	\end{subequations}
\end{lemma}

\begin{proof}
	It directly follows from duality, by letting $y_p$ (for $p \in N$), $y_{p,e}$ (for $p \in N$ and $e \in E$), $y_{p,t}$ (for $p \in N$), and $y_\theta$ (for $\theta \in \Theta$) be the dual variables associated to, respectively, Constraints~\eqref{c:1},~\eqref{c:2},~\eqref{c:3},~and~\eqref{c:4}.
\end{proof}

Since $|A|$ is exponential in the size of the game, Problem~\eqref{LP:1} features exponentially many variables, while its number of constraints is polynomial.
Conversely, Problem~\eqref{LP:2} has polynomially many variables and exponentially many constraints, which enables the use of the ellipsoid algorithm to find an optimal solution to Problem~\eqref{LP:2} in polynomial time.~\footnote{For additional details on how the ellipsoid algorithm can be adopted to solve optimization problems see~\citep{grotschel1981ellipsoid}.}
This requires a polynomial-time separation oracle for Problem~\eqref{LP:2}, \emph{i.e.}, a procedure that, given a vector $y$ of dual variables, it either establishes that $y$ is feasible for Problem~\eqref{LP:2} or, if not, it outputs an hyperplane separating $y$ from the feasible region.
In the following, we focus on a particular type of separation oracles: those generating violated constraints.
%~\footnote{In general, a separation oracle may output an arbitrary separating hyperplane. Here, we consider the special case in which it returns a violated constraint of Problem~\ref{LP:2}.}
%

Given that Problem~\eqref{LP:2} has an exponential number of constraints, a polynomial-time separation oracle is not readily available.
It turns out that, in our setting, we can design one by leveraging the symmetry of the players and the fact that the cost functions are affine, as described in the following.

First, we prove that Problem~\eqref{LP:2} always admits an optimal player-symmetric solution, \emph{i.e.}, a vector $y$ such that, for each pair of players $p, q \in N$, it holds $y_p = y_q$, $y_{p,e}=y_{q,e}$ for all $e \in E$, and $y_{p,t}=y_{q,t}$.
This result allows us to restrict the attention to player-symmetric vectors $y$.
% when designing our polynomial-time separation oracle. 

\begin{lemma}\label{lem:player_symm}
	Problem~\eqref{LP:2} always admits an optimal player-symmetric solution.
\end{lemma}

\begin{proof}
	Given any optimal solution $y$ to Problem~\eqref{LP:2}, we can always recover, in polynomial time, a player-symmetric optimal solution $\tilde{y}$.
	Specifically, for every $p \in N$, let $\tilde{y}_p= \frac{\sum_{p \in N} y_p}{n}$, $\tilde{y}_{p,e}=\frac{\sum_{p\in N} y_{p,e}}{n}$ for all $e \in E$, and $\tilde{y}_{p,t}=\frac{\sum_{p\in N} y_{p,t}}{n}$, while $\tilde{y}_\theta=y_\theta$ for every $\theta \in \Theta$.
	First, notice that $y$ and $\tilde{y}$ provide the same objective value, as $\tilde{y}_\theta=y_\theta$ for all $\theta \in \Theta$.
	Thus, we only need to prove that $\tilde{y}$ satisfies all the constraints of Problem~\eqref{LP:2}.
	For $a \in A$ and $i \in [n]$, let us denote with $\pi_i(a)$ an action profile $a' \in A$ such that $a_p'=a_{((p+i)\mod n)}$, \emph{i.e.}, a permutation of $a$ in which each player $p \in N$ takes on the role of player $(p+i)\mod n$.
	Moreover, let $\pi(a) = \bigcup_{i \in [n]} \pi_i(a)$.
	Constraints~\eqref{d:1} are satisfied by $\tilde{y}$, since, for every $\theta \in \Theta$ and $a \in A$, it holds:
	\begin{align*}
		&\mu_\theta \left( \sum_{p \in N} c_{p,\theta}(a) \tilde{y}_{p} - \sum_{p \in P}\sum_{e \in E} c_{e,\theta}\left(f_e^a+I_{\{e \notin a_p\}} \right) \tilde{y}_{p,e} \right) + \tilde{y}_{\theta} = \\ &= \frac{1}{n} \sum_{a' \in \pi(a)} \mu_\theta \left( \sum_{p \in P} c_{p,\theta}(a') y_p -  \sum_{p \in P}\sum_{e \in E} c_{e,\theta} \left(f_e^{a'}+I_{\{e \notin a_p'\} } \right) y_{p,e} \right) +y_{\theta}  \le \\
		& \le  \frac{1}{n} \sum_{a' \in \pi(a)} \mu_\theta \sum_{p \in N} c_{p,\theta}(a') = \mu_\theta \sum_{p \in N} c_{p,\theta}(a').
	\end{align*}
	Similar arguments show that $\tilde{y}$ satisfies all the other constraints, concluding the proof.
\end{proof}

Notice that any polynomial-time separation oracle for Problem~\eqref{LP:2} can explicitly check whether each member of the polynomially many Constraints~\eqref{d:2},~\eqref{d:3},~and~\eqref{d:4} is satisfied for the given $y$.
%, since they are a polynomial number in the size of the game.
%
Thus, we focus on the separation problem restricted to the exponentially many Constraints~\eqref{d:1}, which, using Lemma~~\ref{lem:player_symm}, can be formulated as stated in the following.

\begin{lemma}
	Given a player-symmetric $y$, solving the separation problem for Constraints~\eqref{d:1} amounts to finding $\theta \in \Theta$ and $a \in A$ that are optimal for the following problem:
	\begin{align}
		&\min_{{\theta \in \Theta, a \in A}}  \mu_\theta \left( (1- \bar y) \sum_{p \in N}  c_{p,\theta}(a) -  \sum_{p \in N} \sum_{e \in E} c_{e,\theta}\left(f_e^a+I_{\{e \notin a_p \}} \right) \bar y_e \right) - y_\theta, \label{Sep}
	\end{align}
	where we let $\bar y = y_p$ and $\bar y_e = y_{p,e}$ for all $e \in E$.
\end{lemma}

Next, we show how Problem~\eqref{Sep} can be equivalently formulated avoiding the minimization over the exponentially-sized set $A$.
Intuitively, we rely on the fact that, for a fixed $\theta \in \Theta$, we can exploit the symmetry of the players to equivalently represent action profiles $a \in A$ as integer vectors $q$ of edge congestions $q_e \in [n]$, for all $e \in E$.

\begin{lemma}\label{lem:separation}
	Problem~\eqref{Sep} can be formulated as $\min_{\theta \in \Theta} \chi(\theta)$, where $\chi(\theta)$ is the optimal value of the following problem:
	 \begin{subequations}\label{Minflow}
		\begin{align}
		\min_{q \in \mathbb{Z}_+^{|E|}} & \,\,\,  (1-\bar{y})\sum_{e \in E} \alpha_{e,\theta} q_e^2+\beta_{e,\theta} q_e - \sum_{e \in E} \bar{y}_e \Big( n \alpha_{e,\theta} q_e + (n-q_e) \alpha_{e,\theta} + n \beta_{e,\theta} \Big) \label{f:0} \\
		& \hspace{-3mm}\sum_{v \in V: e=(s,v) \in E} q_e=n \label{f:1} \\
		& \hspace{-3mm}\sum_{v \in V: e=(v,t) \in E} q_e=n \label{f:2}\\
		& \hspace{-3mm}\sum_{\substack{v' \in V: \\e=(v',v) \in E}} q_e = \sum_{\substack{v' \in V: \\e=(v,v') \in E}} q_e \hspace{2cm} \forall v \in V \setminus \{s,t\} \label{f:3} %\\
		% & \hspace{-3mm}q_e \in \mathbb{Z}_+ \hspace{4.7cm} \forall e \in E.
		\end{align}
	\end{subequations}
\end{lemma}

\begin{proof}
	First, given a state $\theta \in \Theta$, Problem~\eqref{Sep} reduces to computing $\chi(\theta) \defeq \min_{a \in A} (1- \bar y) \sum_{p \in N} c_{p,\theta}(a) - \sum_{p \in N}\sum_{e \in E}  c_{e,\theta} \left( f_e^a+I_{\{e \notin a_p\}} \right) \bar y_e $, where the function to be minimized only depends on the number of players selecting each edge $e \in E$ in $a$, rather than the identity of the players who are choosing $e$ (since they are symmetric).
	Letting $q(e) \in [n]$ be the congestion level of edge $e \in E$ and using $c_{e,\theta} = \alpha_{e,\theta} q_e + \beta_{e,\theta} $ (affine costs), it holds
	$\sum_{p \in N} c_{p,\theta}(a)=\sum_{e \in E} \alpha_{e,\theta} q_e^2+\beta_{e,\theta} q_e$, and, for every $e \in E$, $\sum_{p \in N} c_{e,\theta} \left(f_e^a+I_{\{e \notin a_p\}} \right)= n \alpha_{e,\theta}q_e + (n-q_e) a_{e,\theta} + n \beta_{e,\theta}$.
	This gives Objective~\eqref{f:0}.
	Moreover, Constraints~\eqref{f:1},~\eqref{f:2},~and~\eqref{f:3} ensure that $q$ is well defined.
\end{proof}

Let us remark that computing an optimal integer solution to Problem~\eqref{Minflow} is necessary in order to (possibly) find a violated constraint for the given $y$; otherwise, we would not be able to easily recover an action profile $a \in A$ from a vector $q$.

Now, we show that an optimal integer solution to Problem~\eqref{Minflow} can be found in polynomial time by reducing it to an instance of integer min-cost flow problem.
Intuitively, it is sufficient to consider a modified version of the original graph $G$ in which each edge $e \in E$ is replaced with $n$ parallel edges with unit capacity and increasing unit costs.
This is possible given that the Objective~\eqref{f:0} is a convex function of $q$, which is guaranteed by the assumption that the costs are affine. 

\begin{lemma}\label{lem:min_cost_flow}
	An optimal integer solution to Problem~\eqref{Minflow} can be found in polynomial time by solving a suitably defined instance of integer min-cost flow problem.
\end{lemma}

\begin{proof}
	First, notice that Objective~\eqref{f:0} is a sum edge costs, in which the cost of each edge $e \in E$ is a convex function of the edge congestion $q_e$, as the only quadratic term is $(1-\bar y) a_{e,\theta} q_e^2$, where the multiplying coefficient is always positive, given $\bar y \leq 0$ and $\alpha_{e,\theta} \geq 0$.
	This allows us to formulate Problem~\eqref{Minflow} as an instance of integer min-cost flow problem.
	We build a new graph where each $e \in E$ is replaced with $n$ parallel edges, say $e_i$ for $i \in [n]$.
	For $e \in E$ and $i \in [n]$, let us define $g(e,i) \defeq (1-\bar{y}) \left( \alpha_{e,\theta} i^2+\beta_e i \right) -\bar{y}_e \left( n \alpha_{e,\theta} i + (n-i) \alpha_{e,\theta} + n \beta_{e,\theta} \right)$.
	%
	%Moreover, let $g(e,0) = 0$.
	%
	Each (new) edge $e_i$ has unit capacity and a per-unit cost equal to $\delta(e_i) \defeq g(e,i) - g(e,i-1)$.
	Clearly, finding an integer min-cost flow is equivalent to minimizing Objective~\eqref{f:0}.
	Notice that, since the original edge costs are convex, it holds $\delta(e_i) \geq \delta(e_j)$ for all $j < i \in  [n]$.
	Thus, an edge $e_i$ is used (\emph{i.e.}, it carries a unit of flow) only if all the edges $e_j$, for $j < i \in [n]$, are already used.
	This allows us to recover an integer vector $q$ from a solution to the min-cost flow problem.
	Finally, let us recall that we can find an optimal solution to the integer min-cost flow problem in polynomial time by solving its LP relaxation.
\end{proof}

The last lemma allows us to prove our main result:

\begin{theorem}
	Given a symmetric BNCG, an optimal \emph{ex-ante} persuasive signaling scheme can be computed in poly-time.
\end{theorem}

\begin{proof}
	The algorithm applies the ellipsoid algorithm to Problem~\eqref{LP:2}.
	At each iteration, we require that the vector of dual variables $y$ given to the separation oracle be player-symmetric, which can be easily obtained by applying the symmetrization technique introduced in the proof of Lemma~\ref{lem:player_symm}.
	% to build a symmetric instance of Problem~\ref{LP:2}.
	%
	The separation oracle needs to solve an instance of integer min-cost flow problem for every $\theta \in \Theta$ (see Lemmas~\ref{lem:separation}~and~\ref{lem:min_cost_flow}).
	Notice that an integer solution is required in order to be able to identify a violated constraint.
	Finally, the polynomially many violated constraints generated by the ellipsoid algorithm can be used to compute an optimal $\phi$.
\end{proof}

% \textbf{proof overview.} In an equilibrium with low social cost only players $p_i$ must be on the resource $r_t$ with high probability and ,thus, all players $p_v$ are playing on resources $r_v$ or $r_{\bar v}$. However, since they don't deviate on $r_t$, they are alone with high probability. This implies that there exists an action profile (played with positive probability) in which all players $r_v$ are alone. If $p_v$ plays on $r_{\bar v}$, $v$ is true, if he plays $r_{v}$, $v$ is false.
% Since all $r_t$ are alone, all players $r_{\phi,q}$ play on the resource of a true literal, implying that 3SAT  is satisfiable.
% For prove the existence of an equilibrium with low social cost when 3SAT is satisfiable, we will need some other players and the use of correlated strategies. For example, we introduce players $r_{v,1},r_{v,2}$ to avoid that clause players deviates on a resource with only player $r_v$.

%\clearpage

\section{The Curse of Asymmetry}\label{sec:hard}

In this section, we provide our hardness result on asymmetric BNCGs.
Our proof is split into two intermediate steps: \emph{(i)} we prove the hardness for a simple class of asymmetric non-Bayesian congestion games in which each player selects only one resource (Lemma~\ref{lm:hard}), and \emph{(ii)} we show that such games can be represented as NCGs with only a polynomial blow-up in the representation size (Lemma~\ref{lm:2}).
Our main result reads:

\begin{theorem}\label{thm:hardness}
	The problem of computing an optimal \emph{ex ante} persuasive signaling scheme in BNCGs with asymmetric players is \NPHard, even with affine costs.~\footnote{Without affine costs, computing an optimal \emph{ex ante} persuasive signaling scheme is trivially \NPHard~even in symmetric BNCGs. This directly follows from~\cite{meyers2012complexity}, which show that even finding an optimal action profile (that is also an optimal Nash equilibrium) is \NPHard~in symmetric (non-Bayesian) NCGs.}
\end{theorem}

The proof of Theorem~\ref{thm:hardness} is based on a reduction which maps an instance of 3SAT (a well-known \NPHard~problem, see~\citep{garey1979computers}) to a game in the class of \emph{singleton congestion games} (SCGs)~\citep{ieong2005fast}.
%
%, where each player can select only one resource at a time from the subset of resources she has available.
%
A (non-Bayesian) SCG is described by a tuple $(N,R,\{A_p\}_{p\in N},\{c_r\}_{r\in R})$, where $R$ is a finite set of resources, each player $p \in N$ selects a single resource from the set $A_p \subseteq R$ of available resources, and resource $r \in R$ has a cost $c_r : \mathbb{N} \to \mathbb{R}_+$.
%, defined as a function of the congestion.
%
Another way of interpreting SCGs is as games played on parallel-link graphs, where each player can select only a subset of the edges.

First, let us provide the following definition and notation.

\begin{definition}[3SAT]
	Given a finite set $C$ of three-literal clauses defined over a finite set $V$ of variables, is there a truth assignment to the variables satisfying all the clauses?
\end{definition}

We denote with $l \in \varphi$ a literal (\emph{i.e.}, a variable or its negation) appearing in a clause $\varphi \in C$. %and $v(l) \in V$ denote the variable corresponding to that literal.
Moreover, we let $m$ and $s$ be, respectively, the number of clauses and variables, \emph{i.e.}, $m \coloneqq |C|$ and $s \coloneqq |V|$. W.l.o.g., we assume that $m \ge s$.

Lemma~\ref{lm:hard} introduces our main reduction, proving that finding a social-cost-minimizing CCE is \NPHard~in SCGs with asymmetric players, \emph{i.e.}, whenever the resource sets $A_p$ are different among each other.~\footnote{The reduction in Lemma~\ref{lm:hard} does \emph{not} rely on standard constructions, as most of the reductions for congestion games only work with action profiles, while ours needs randomization. Indeed, in asymmetric SCGs, a social-cost-minimizing action profile can be computed in poly-time by solving a suitable instance of min-cost flow. This also prevents the use of other techniques for proving the hardness of CCEs, such as, \emph{e.g.}, those by~\cite{barman15cce}.}
Notice that the games used in the reduction are \emph{not} Bayesian; this shows that the hardness fundamentally resides in the asymmetry of the players.

\begin{lemma}\label{lm:hard}
	The problem of computing a social-cost-minimizing CCE in SCGs with asymmetric players is \NPHard, even with affine costs.
\end{lemma}
\begin{proof}
	Our 3SAT reduction shows that the existence of a polynomial-time algorithm for computing a social-cost-minimizing CCE in SCGs would allow us to solve any 3SAT instance in polynomial time.
	Given $(C,V)$, let $z \coloneqq m^{30}$, $u \coloneqq m^{12}$, and $\epsilon \coloneqq \frac{1}{m^4}$. %and $h=\frac{1}{m^{10}}$.
	We build an SCG $\Gamma(C,V)$ admitting a CCE with social cost smaller than or equal to  $\gamma \coloneqq z^2+ (4us+s+3m)(z-u) + \frac{3z}{m^{9}}$ iff $(C,V)$ is satisfiable.
	
	\textbf{Mapping.}
	$\Gamma(C,V)$ is defined as follows (for every $r \in R$, the cost $c_r$ is an affine function with coefficients $\alpha_r$ and $\beta_r$).
	\begin{itemize}[nolistsep,itemsep=1mm]
		\item $N= \{ p_{v} \mid v  \in V\} \cup \{ p_{\varphi,q} \mid \varphi \in C, q \in [3]\} \cup \{p_{v,j}, p_{\bar v,j} \mid v \in V, j \in [2u]\} \cup \{ p_{i} \mid i  \in [z]\}$;
		\item $R = \{ r_t\} \cup \{ r_v, r_{\bar v}, r_{v,1}, r_{v,2}, r_{\bar v,1},r_{\bar v,2} \mid v \in V\}$;
		\item $A_{p_v} = \{r_v, r_{\bar v}, r_{t} \}\ \ \forall v \in V$;
		\item $A_{p_{\varphi,q}} = \{ r_l \mid l \in \varphi\}  \ \ \forall\ \varphi \in C, \forall q \in [3]$;
		\item $A_{p_{v,j}}= \{ r_v, r_{v,1},r_{ v,2} \}\ \ \forall v \in V, \forall j \in [2u]$;
		\item $A_{p_{\bar v,j}}= \{ r_{\bar v}, r_{\bar v,1},r_{\bar v,2} \} \ \ \forall v \in V, \forall j \in [2u]$;
		\item $A_{p_i} = \{ r_t \} \ \ \forall i \in [z]$;
		\item $\alpha_{r_v}=\alpha_{r_{\bar v}}=\epsilon$ and $\beta_{r_v}=\beta_{r_{\bar v}}=z+1-\epsilon\ \ \forall v \in V$;
		\item $\alpha_{r_{v,1}}=\alpha_{r_{v,2}}=\alpha_{r_{\bar v,1}}=\alpha_{r_{\bar v,2}}=1\ \ \forall v \in V$;
		\item $\beta_{r_{v,1}}= \beta_{r_{v,2}}=\beta_{r_{\bar v,1}}=\beta_{r_{\bar v,2}}=z+1-u\ \ \forall\ v \in V$;
		\item $\alpha_{r_t}=1$ and $\beta_{r_t}=0$;
	\end{itemize}

	% Notice that, in a game $\Gamma(C,V)$, all the players in the set $\{ p_{i} \mid i  \in [z]\} $ must select resource $r_t$.
	%, these players are playing $r_t$ with probability $1$.
	
	%\begin{figure*}[!htp]
	%	\centering
	%	\includegraphics[width=\textwidth]{figs/reduction_opt.pdf}
	%	\caption{Example of a game instance $\Gamma_\epsilon(C,V)$ used in the reduction in the proof of Theorem~\ref{thm:np_hard_opt}, with $V=\{x,y,z\}$, $C = \{\varphi_1,\varphi_2\}$, $\varphi_1 = x \vee y \vee z$, and $\varphi_2 = \bar x \vee y \vee \bar z$.}
	%	\label{fig:reduction_opt}
	%\end{figure*}
	%

	\textbf{If.} 
	Suppose $(C,V)$ is satisfiable, and let $\tau : V \to \{\mathsf T,\mathsf F\}$ be a truth assignment satisfying all the clauses in $C$. 
	Using $\tau$, we recover a CCE $\phi \in \Delta_A$ with social cost smaller than or equal to $\gamma$, having in its support the action profiles $a^{k} =(a_p^k)_{p\in N} \in \bigtimes_{p \in N} A_p$ for $k \in [3]$, defined as follows.
	%
	% For every clause $\varphi \in C$, there exists a literal $l \in \varphi$ with $\tau(l)= \mathsf T$.
	%
	First, let us consider a congestion game $\Gamma_\textnormal{R}$ restricted to the players in $\{ p_{\varphi,q} \mid \varphi \in C, q \in [3] \}$, with action spaces limited to the resources $r_l \in A_{p_{\varphi,q}}$ with $\tau(l)= \mathsf T$ (since $\tau$ satisfies all clauses, each player has at least one action).
	Clearly, $\Gamma_\textnormal{R}$ admits a pure NE~\citep{rosenthal1973class}.
	Moreover, we show that, in any pure NE, each resource is selected by at least one player.
	By contradiction, suppose that there exists a resource $r_l$ such that no player $p_{\varphi,q}$ chooses it.
	Then, there must be at least two players $p_{\varphi,q}$ (with $l \in \varphi$) selecting a resource different from $r_l$.
	It is easy to check that, then, there must be one player with an incentive to deviate, contradicting the NE assumption.
	%
	% Take a clause that includes $l$. There are no players on this resource and at least two players $r_{\varphi,q}$ on another resource $r_{l'}$, $l' \in \varphi$. This implies that one of the players has a incentive to deviate.
	%
	For every $\varphi \in C$ and $ q \in [3]$, we let $a^1_{p_{\varphi,q}}$ (respectively, $a^2_{p_{\varphi,q}}$ and $a^3_{p_{\varphi,q}}$) be equal to the resource played by the corresponding player in a pure NE of $\Gamma_\textnormal{R}$.
	For every literal $l \in \{v, \bar v \mid v \in V  \}$ and $j \in [u]$, we let $a^1_{p_{l,j}}= r_{l,1}$, $a^2_{p_{l,j}}= r_{l,2}$, and $a^3_{p_{l,j}}= r_{l}$ if $\tau(l)= \mathsf F$, while $a^3_{p_{l,j}}= r_{l,1}$ if $\tau(l)= \mathsf T$.
	Similarly, for $j \in [2u]: j > u$, we let $a^1_{p_{l,j}}= r_{ l,2}$, $a^2_{p_{l,j}}= r_{l,1}$, and $a^3_{p_{l,j}}= r_{l}$ if $\tau(l)= \mathsf F$, while $a^3_{p_{l,j}}= r_{l,2}$ otherwise.
	For every $v \in V$, we let $a^3_{p_v}=r_t$ and $a^1_{p_v} =a^2_{p_v}=r_{ v}$ if $\tau(v)= \mathsf F$, while $a^1_{p_v} =a^2_{p_v}=r_{ \bar v}$ if not.
	Finally, we let $\phi_{a^1} = \phi_{a^2} = \frac{1}{2}-\frac{1}{2m^{10}}$, while $\phi_{a^3} = \frac{1}{m^{10}}$.
	We show that players have no incentive to defect from $\phi$.
	Given that player $p_{\varphi, q}$'s action (for $\varphi \in C$ and $q \in [3]$) is determined by a pure NE of $\Gamma_\textnormal{R}$, she does not have any incentive to select a resource $r_l \in A_{\varphi, q}$ with $\tau(l)=\mathsf T$ (as it is not selected by other players).
	Moreover, in $\phi$, player $p_{\varphi,q}$'s expected cost is at most $z+1+3 \epsilon m$, while she would pay at least $(z+1+\epsilon) (1-\frac{1}{m^{10}})+ (z+1+2u\epsilon) \frac{1}{m^{10}}=z+1+2 \epsilon m^2$ by selecting a resource $r_l \in A_{\varphi, q}$ with $\tau(l)= \mathsf F$.
	Each player $p_v$ (for $v \in V$) does not defect from $\phi$, since her expected cost is $(z+1)(1-\frac{1}{m^{10}})+ (z+s) \frac{1}{m^{10}}$, while she would pay: \emph{(i)} the same by switching to $r_t$, \emph{(ii)} at least $(z+1+\epsilon)$ by playing $r_l$ with $l \in \{v,\bar v\}$ and $\tau(l)=\mathsf T$ (as there is at least one player $p_{\varphi,q}$ on $r_l$), or \emph{(iii)} at least $(z+1)(1-\frac{1}{m^{10}})+(z+1+2u \epsilon)\frac{1}{m^{10}}= z+1+2 \frac{1}{m^{2}}$ by selecting $r_l$ with $l \in \{v,\bar v\}$ and $\tau(l)=\mathsf F$.
	%
	% % % % Instead, by deviating to $r_l$ with $l \in \{v,\bar v\}$, her cost would be at least of $(z+1+\epsilon)$ (if $\tau(l)=\mathsf T$, as there is at least a player $p_{\varphi,q}$ on it) or $(z+1)(1-m^{-10})+(z+1+2u \epsilon)m^{-10}= z+1+2m^{-2} $ (if $\tau(l)=\mathsf F$).
	%
	%All players $p_{v}$ will not deviate, as they are paying $(z+1)(1-m^{-10})+ [z+s]m^{-10}$ and deviating on $r_{v_t}$ will pay the same. 
	%
	%Moreover, they would not deviate on a resource $r_l$ with $l \in \{v,\bar v\}$: if $\tau(l)=\mathsf T$ they will pay at least $(z+1+\epsilon)$ (there is at least a player $p_{\varphi,q}$ on it), if $\tau(l)=\mathsf F$ they will pay at least $(z+1)(1-m^{-10})+(z+1+2u \epsilon)m^{-10}= z+1+2m^{-2} $.
	%
	Each player $p_{l,j}$ (for $l \in \{ v, \bar v \mid v \in V \}$ and $j \in [2u]$) with $\tau(l)=\mathsf{F}$ does not deviate, since her cost is $(z+1)(1-\frac{1}{m^{10}})+ (z+1 -\epsilon +2u \epsilon) \frac{1}{m^{10}}$, while she would pay: \emph{(i)} at least $(z+1)(\frac{1}{2}-\frac{1}{2 m^{10}})+(z+2) (\frac{1}{2}-\frac{1}{2 m^{10}})$ by switching to either $r_{l,1}$ or $r_{l,2}$, or \emph{(ii)} at least $(z+1+\epsilon)(1-\frac{1}{m^{10}})+(z+1-\epsilon+2u \epsilon)\frac{1}{m^{10}}$ by selecting resource $r_l \in A_{p_{l,j}}$.  
	%
	% and deviating on $r_{l,1}$ or $r_{l,2}$ will pay at least $(z+1)(1/2-m^{-10}/2)+(z+2) (1/2-m^{-10}/2)$. Moreover, they will not deviate on $r_{l}$ since they will pay at least $(z+1+\epsilon)(1-m^{-10})+(z+1-\epsilon+2u \epsilon)m^{-10}$.  
	%
	Moreover, each player $p_{l,j}$ with $\tau(l)=\mathsf{T}$ does not deviate either, as her cost is $(z+1)$, while she would pay: \emph{(i)} at least $(z+1+\epsilon)$ by playing $r_l$, or \emph{(ii)} at least $(z+1)(\frac{1}{2}+\frac{1}{2 m^{10}})+(z+2) (\frac{1}{2}-\frac{1}{2 m^{10}})$ by switching to either $r_{l,1}$ or $r_{l,2}$.
	Finally, the CCE provides a social cost smaller than or equal to  $(z^2+(4us+s+3m)(z+u))(1-\frac{1}{m^{10}})+((z+s)^2+(4us+m)(z+u)) \frac{1}{m^{10}}\le \gamma$, where the last inequality comes from  $2u(4us+s+3m)+\frac{s^2}{m^{10}}\le \frac{z}{m^{9}}$.

	\textbf{Only if.}
	Suppose there exists a CCE $\phi \in \Delta_A$ with social cost smaller than or equal to $\gamma$.
	First, we prove that, with probability at most $\frac{1}{m^8}$, at least one player $p_v$ plays $r_t$.
	By contradiction, assume that this is not the case.
	Then, the social cost would be at least $(z^2+(4us+s+3m)(z-u))(1-\frac{1}{m^8}) + ((z+1)^2+ (4us+s+3m-1)(z-u)) \frac{1}{m^{8}} \ge z^2+(4us+s+3m) (z-u)+  (2z-z) \frac{1}{m^{8}} > \gamma$.
	This implies that each player $p_v$ is playing either $r_v$ or $r_{\bar v}$ with probability at least $1-\frac{1}{m^8}$.
	We prove that $p_v$ is the only player on that resource with probability at least $1-\frac{1}{m^8}-\frac{1}{m^2}$.
	Otherwise, by contradiction, her cost would be at least $z+1+  \frac{\epsilon}{m^2}=z+1+ \frac{1}{m^6}$, while by playing $r_t$ she would pay at most $(z+1)(1-\frac{1}{m^8})+ (z+s)\frac{1}{m^8}\le z+1 + \frac{1}{m^7}$.
	By a union bound, there exists an action profile $a = (a_p)_{p \in N} \in \bigtimes_{p \in N} A_p$ played with probability at least $1- s ( \frac{1}{m^8}+\frac{1}{m^2}) >0$ in which all the players $p_v$ are alone on their resources (either $r_v$ or $r_{\bar v}$).
	Let $\tau: V \to \{ \mathsf{T}, \mathsf{F} \}$ be a truth assignment such that $\tau(v)=\mathsf T$ if $a_{p_v}=r_{\bar v}$ and $\tau(v)=\mathsf F$ if $a_{p_v}=r_{ v}$.
	Then, $\tau$ satisfies all the clauses, since all the players $p_{\varphi,q}$ play $r_l$ with $\tau(l) = \mathsf{T}$ and, thus, all the clauses have at least one true literal.
\end{proof}

The following lemma concludes the proof of Theorem~\ref{thm:hardness}.
%
%We conclude by proving that any SCG can be equivalently expressed as an NCG with only a polynomial increase in the representation size.

\begin{lemma}\label{lm:2}
	Any SCG can be represented as an NCG of size polynomial in the size of the original SCG.
\end{lemma}

\begin{proof}
	Given an SCG $(N,R,\{A_p\}_{p\in N},\{c_r\}_{r\in R})$ we build an NCG $(N,G,\{c_e\}_{e \in E},\{(s_p,t_p)\}_{p \in N})$, as follows.
	The graph $G=(V,E)$ has two nodes $v_{r,1}, v_{r,2} \in V$ for each resource $r \in R$, and, additionally, for every player $p \in N$, there is a source node $s_p \in V$ and a destination one $t_p \in V$.
	Moreover, there is an edge $(v_{r,1}, v_{r,2}) \in E$ for every $r \in R$ and, for every $p \in N$ and $r \in A_p$, there two edges $(s_p, v_{r,1}) \in E$ and $(v_{r,2},t_p) \in E$.
	Finally, for the edges $e = (v_{r,1}, v_{r,2})$, we let $c_{e}=c_r$, while $c_e=0$ for all the other edges. 
\end{proof}

\section{Discussion and Future Works}\label{sec:discussion}

The paper studies information-structure design problems in atomic BNCGs, where an informed sender can observe the actual state of the network and commit to a signaling scheme.
We focus on the problem of computing optimal {\em ex ante} persuasive signaling schemes in such setting.
We show that, with affine costs, {\em simmetry} is the property marking the transition from polynomial-time tractability to \textsf{NP}-hardness.

%
%\red{
%In order to prove our main positive result, we provide a polynomial-time separation oracle based on a suitably defined min-cost flow problem.
%%
%Moreover, our hardness result holds even in non-Bayesian singleton congestion games with affine costs, which is arguably the simplest class of congestion games where asymmetry can be introduce.
%}

In the future, we are interested in studying the problem of approximating optimal {\em ex ante} persuasive signaling schemes, and in the design of practical algorithms for real-world network signaling problems.
Moreover, in order to make the framework even more applicable, it would be interesting to explore how the sender can handle uncertainty about receivers' payoffs, and to be robust to mismatching priors.

\bibliographystyle{named}
\bibliography{refs}

\end{document}